\documentclass[preprint,3p]{elsarticle}

\usepackage{amsmath,amssymb,array}
\usepackage{amsthm}
\usepackage{microtype}
\usepackage[utf8]{inputenc}
\usepackage{graphicx}

\usepackage{tikz}
\usetikzlibrary{%
  positioning,
  decorations.pathreplacing,
  decorations.pathmorphing
}

\newcommand{\SQD}[2]{\ensuremath{d_{#1(#2)}}}

\newcommand{\DUO}[2]{\ensuremath{(#1[#2],#1[#2+1])}}

\newtheorem{teo}{Theorem}
\newtheorem{lemma}{Lemma}

\newtheorem{Cl}{Claim}

\theoremstyle{definition}
\newtheorem{Def}{Definition}[section]

\begin{document}

\title{Parameterized Tractability of the Maximum-Duo Preservation
  String Mapping Problem}
\author[itbcnr,disco]{Stefano Beretta\corref{cor1}}
\ead{stefano.beretta@disco.unimib.it}

\author[isegi]{Mauro Castelli}
\ead{mcastelli@novaims.unl.pt}

\author[unibg]{Riccardo Dondi}
\ead{riccardo.dondi@unibg.it}
\cortext[cor1]{Corresponding author}

\address[itbcnr]{Istituto di Tecnologie Biomediche, Consiglio
  Nazionale delle Ricerche, Segrate - Italia}
\address[disco]{Dipartimento di Informatica, Sistemistica e
  Comunicazione, Universit\`a degli Studi di Milano - Bicocca, Milano - Italia}
\address[isegi]{NOVA IMS, Universidade Nova de Lisboa, Lisboa - Portugal}
\address[unibg]{Dipartimento di Scienze Umane e Sociali, Universit\`a degli
  Studi di Bergamo, Bergamo - Italia}

\begin{abstract}
In this paper we investigate the parameterized complexity of the
Maximum-Duo Preservation String Mapping Problem, the complementary of
the Minimum Common String Partition Problem. We show that this problem is
fixed-parameter tractable when parameterized by the number $k$ of
conserved duos, by first giving a parameterized algorithm based on the
color-coding technique and then presenting a reduction to a kernel of
size $O(k^6)$.
\end{abstract}

\maketitle

\section{Introduction}
Minimum Common String Partition (MCSP) is a problem emerged in the field of
comparative genomics~\cite{DBLP:journals/tcbb/ChenZFNZLJ05} and, in particular,
in the context of ortholog gene
assignments~\cite{DBLP:journals/tcbb/ChenZFNZLJ05}.
Given two strings (genomes) $A$ and $B$, MCSP asks for a partition
of the two strings into a minimum cardinality multiset of identical
substrings. 
The complexity of this problem has been previously studied in literature.
More precisely, the MCSP problem is known to be APX-hard, even when each
symbol has at most $2$ occurrences in each input
string~\cite{DBLP:journals/combinatorics/GoldsteinKZ05}, while it
admits a polynomial time algorithm when each symbol occurs exactly
once in each input string.
Approximation algorithms for this problem have been proposed
in~\cite{DBLP:journals/talg/ChrobakKS05,DBLP:journals/talg/CormodeM07,DBLP:journals/combinatorics/GoldsteinKZ05,DBLP:journals/combinatorics/KolmanW07}.
More precisely, an $O(\log n \log^*n)$-approximation algorithm has been given
in~\cite{DBLP:journals/talg/CormodeM07}, while an $O(k)$-approximation
algorithm, when the number of occurrences of each symbol is bounded by $k$,
has been given in~\cite{DBLP:journals/combinatorics/KolmanW07}.
When $k=2$ and $k=3$ respectively, approximation algorithms of factor $1.1037$
and $4$, respectively, have been given
in~\cite{DBLP:journals/combinatorics/GoldsteinKZ05}.
 
The parameterized complexity~\cite{DowneyFellows,niedermeier} of MCSP has also
been investigated.
First, fixed-parameter algorithms have been given when the problem is
parameterized by two parameters.
In~\cite{DBLP:conf/wabi/Damaschke08} this problem has been shown to be
fixed-parameter tractable, when parameterized by the number of substrings in
the solution and by the repetition number of the input strings.
Then, the MCSP problem has been shown to be fixed-parameter tractable when
parameterized by the number of substrings in the solution and the maximum
number of occurrences of a symbol in an input
string~\cite{DBLP:conf/wabi/BulteauFKR13,DBLP:journals/jco/JiangZZZ12}.
Recently, MCSP has been shown to be fixed-parameter tractable when
parameterized by the single parameter number of substrings of the
partition~\cite{DBLP:conf/soda/BulteauK14}.

Here, we consider the complementary of the MCSP problem, called Maximum-Duo
Preservation String Mapping Problem, where instead of minimizing the number
of identical substrings in the partition, we aim to maximize the number of
preserved duos, that is, the number of adjacencies of symbols that are not
broken by the partition.
This problem has been proposed in~\cite{DBLP:journals/tcs/ChenCSPWT14},
has been shown to be APX-hard when each symbol has at most $2$ occurrences
in each input string~\cite{DBLP:conf/wabi/BoriaKLM14}, and can be approximated
within factor $\frac{1}{4}$~\cite{DBLP:conf/wabi/BoriaKLM14}.

In this work, we study the parameterized complexity of the Maximum-Duo
Preservation String Mapping Problem, where the parameter is the number of
preserved adjacencies (duos).
More precisely, after introducing preliminary definitions and properties
of Maximum-Duo Preservation String Mapping in Section~\ref{sec:prel},
we describe in Section~\ref{sec:ftp1} a fixed-parameter algorithm for the
problem, based on the color-coding technique.
Then, in Section~\ref{sec:kernel}, we present a reduction to a polynomial
kernel of size $O(k^6)$.

The results described in this paper are mainly of theoretical interest,
since a solution of the Maximum-Duo Preservation String Mapping Problem
is expected to preserve many adjacencies.
However, the fixed-parameter algorithms we propose can be of interests for
describing the whole parameterized complexity status of the Minimum Common
String Partition Problem and its variants.
For example, while it is still unknown whether Minimum Common String Partition
Problem admits a polynomial kernel, the result in Section~\ref{sec:kernel}
shows that such a kernel exists for the complementary problem.

\section{Preliminaries}
\label{sec:prel}
In this section, we introduce some concepts that will be used in the rest of
the paper and we give the formal definition of the Maximum-Duo Preservation
String Mapping Problem.
Fig.~\ref{fig:example} illustrates some of the definitions we give in this
section.

Let $\Sigma$ be a non-empty finite set of symbols.
Given a string $A$ over $\Sigma$, we denote by $|A|$ the length of $A$
and by $A[i]$, with $1\leq i\leq |A|$ the symbol of $A$ at position $i$.
Moreover, we denote by $A[i,j]$, with $1\leq i < j \leq |A|$, the substring of
$A$ starting at position $i$ and ending at position $j$.
Given a string $A$, a \emph{duo} is an ordered pair of consecutive elements
\DUO{A}{i}.
Consider a duo \DUO{A}{i} in a string $A$ and a duo \DUO{B}{j} in a string $B$;
they are \emph{preservable} if $A[i]=B[j]$ and $A[i+1]=B[j+1]$.

Given two strings $A$ and $B$, such that $B$ is a permutation of $A$,
we say that $A$ and $B$ are \emph{related}. In the rest of the paper we assume
that $|A|=|B|=n$.

Given two related strings $A$ and $B$, a \emph{mapping} $m$ of $A$ into $B$ is
a bijective function from the positions of $A$ to the positions of $B$ such
that $m(i)=j$ implies that $A[i]=B[j]$, i.e. the two positions $i$, $j$ of the
two strings contain the same symbol.
A \emph{partial mapping} $m$ of $A$ into $B$ is a bijective function from a
subset of positions of $A$ to a subset of positions of $B$ such that $m(i)=j$
implies that $A[i]=B[j]$.
The definition of mapping and partial mapping can be extended to two sets of
duos of related  strings $A$ and $B$, that is if positions $i$ and $i+1$, with
$1 \leq i \leq n-1$, are mapped into positions $j$ and $j+1$, with
$1 \leq j \leq n-1$, we say that duo \DUO{A}{i} is mapped into duo \DUO{B}{j}.

Given two related strings $A$ and $B$, and a mapping $m$ of the positions of
$A$ into the positions of $B$, a duo \DUO{A}{i} is preserved if $m(i)=j$ and
$m(i+1)=j+1$ (see Figure~\ref{fig:example} for an example).

\begin{figure}[ht]
\centering
\begin{tikzpicture}
\node(A) {$A=$};
\node[right = 1mm of A] (A1) {$a$};
\node[right = 1mm of A1] (A2) {$b$};
\node[right = 1mm of A2] (A3) {$c$};
\node[right = 1mm of A3] (A4) {$a$};
\node[right = 1mm of A4] (A5) {$b$};
\node[right = 1mm of A5] (A6) {$b$};
\node[right = 1mm of A6] (A7) {$c$};

\node[below = 10mm of A] (B) {$B=$};
\node[right = 1mm of B] (B1) {$a$};
\node[right = 1mm of B1] (B2) {$c$};
\node[right = 1mm of B2] (B3) {$b$};
\node[right = 1mm of B3] (B4) {$b$};
\node[right = 1mm of B4] (B5) {$c$};
\node[right = 1mm of B5] (B6) {$a$};
\node[right = 1mm of B6] (B7) {$b$};

\draw[decorate,decoration={brace, raise=1.7mm, mirror, amplitude=2mm},
thick] (A1.west) -- (A2.east) node[midway, yshift={-2mm}] (A12){};
\draw[decorate,decoration={brace, raise=1.7mm, mirror, amplitude=2mm},
thick] (A5.west) -- (A7.east) node[midway, yshift={-2mm}] (A567){};
\draw[decorate,decoration={brace, raise=1.7mm, amplitude=2mm},
thick] (B6.west) -- (B7.east) node[midway, yshift={2mm}] (B67){};
\draw[decorate,decoration={brace, raise=1.7mm, amplitude=2mm},
thick] (B3.west) -- (B5.east) node[midway, yshift={2mm}] (B345){};;

\draw (A12) edge[out=270, in=90, -, thick] (B67);
\draw (A567) edge[out=270, in=90, -, thick] (B345);
\draw (A3) edge[out=270, in=90, -, thick] (B2);
\draw (A4) edge[out=270, in=90, -, thick] (B1);
\end{tikzpicture}

\caption{An example of two related strings $A$ and $B$.
The mapping of their positions is represented by connecting
positions/substrings.
Position $1$ and $2$ of $A$ are mapped into positions $6$ and $7$ of $B$, hence
duo $(A[1],A[2])$  of $A$ is preserved; position $1$ in $A$ induces duo
$(A[1],A[2])$.
Similarly, the sequence \SQD{A}{5,7} of consecutive duos is mapped into the
sequence \SQD{B}{3,5} of consecutive duos; hence duos $(A[5],A[6])$,
$(A[6],A[7])$ of $A$ are preserved.
The number of preserved duos induced by the mapping is $3$.}
\label{fig:example}
\end{figure}
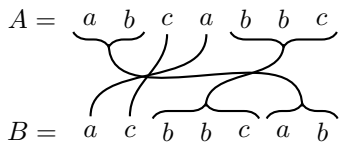

Now, we give the definition of the Maximum-Duo Preservation String Mapping
Problem (in its decision version).

\vspace{1em}
\noindent

{\bf Maximum-Duo Preservation String Mapping Problem (Max-Duo PSM)}

\noindent
\emph{Input:} two related strings $A$ and $B$, an integer $k$.\\
\emph{Output:} is there a mapping $m$ of $A$ into $B$ such that the number of
preserved duos is at least $k$?\\

In this paper, we focus on the parameterized complexity of Max-Duo PSM, when
parameterized by the number $k$ of preserved duos.

Consider a string $S$, with $S \in \{ A,B \}$, and a string $\bar{S}
\in \{ A,B \} \setminus \{ S \}$. 
Given two positions $1 \leq i < j \leq n$, we denote by \SQD{S}{i,j} the
sequence of \emph{consecutive duos} \DUO{S}{i}, \dots , $(S[j-1],S[j])$;
the length of  \SQD{S}{i,j} is the number $j-i$ of consecutive duos in it.
Given the sequence \SQD{S}{i,j} of consecutive duos, the \emph{string
corresponding} to \SQD{S}{i,j} is $S[i,j]$.
Given a string $S[i,j]$ the sequence of duos \emph{induced} by $S[i,j]$ is
\SQD{S}{i,j}.

By a slight abuse of notation, we say that position $i$, with $1 \leq i
\leq n-1$, of a string $S$, induces duo \DUO{S}{i}.

\subsection*{Parameterized Complexity}
We briefly overview the main concepts about parameterized complexity that will
be useful in the rest of paper.
We refer the reader to~\cite{DowneyFellows,niedermeier} for an introduction to
parameterized complexity.

A decision problem $\Pi$ is fixed-parameter tractable under a parameter $k$
when there exists an algorithm of time complexity $O(f(k)poly-time(q))$,
where $q$ is the size of an instance of problem $\Pi$, and $f(k)$ is a computable
function that depends only on $k$ (and not on $q$).

A reduction to a kernel for a given parameterized problem $\Pi$ parameterized
by $k$ is a polynomial-time algorithm that, starting from an instance $(I,k)$
(where $k$ is the parameter) of $\Pi$, computes an instance $(I',k')$ (called
\emph{kernel}) such that $k' \leq k$ and the size of $I'$ is a function on $k$.
It is well-known~\cite{DowneyFellows,niedermeier} that any problem which is
fixed-parameter tractable can be reduced to a kernel of exponential size.
Moreover, there exist problems that admit a polynomial-size kernel, that is the
size of $I'$ is a polynomial function in $k$.

Our first FPT-algorithm is based on the well-known \emph{color-coding} technique, introduced in~\cite{ColorCoding}.
Our color-coding approach is based on the definition of perfect family of hash
functions, used in~\cite{ColorCoding} to derandomize the technique.
Color-coding is a technique widely used to design fixed-parameter algorithms.
While most of the applications of such technique are for graph
problems~\cite{ColorCoding,DBLP:journals/jcss/FellowsFHV11}, it has been
recently applied to problems on
strings~\cite{DBLP:journals/ipl/BonizzoniVDP10,DBLP:journals/alg/Swap2014,DBLP:journals/tcs/BulteauCD15}.

Before introducing the formal definition of perfect hash functions, we give an
informal description of the color-coding technique.
Informally, consider a problem $\Pi$ that, given a set $U$ of $n$ elements, aims
to identify whether there exists a feasible solution which is a subset of $U$ of
size $k$ ($k < n$).
The existence of such a subset can be computed by enumerating
all the subsets of size $k$ in time $O(n^k)$.
However, for some combinatorial problems, color-coding can be used to compute
whether such a set exists or not in time $f(k)poly-time(n)$, where $f(k)$
is a function that depends only on $k$, by first appropriately color-coding the
elements of the set $U$ with $k$ colors and then applying dynamic programming.

Now, we give the formal definition of perfect families of hash
functions, on which the color-coding technique is based.

\begin{Def}
\label{def:perfect-hash}
Let $U$ be a set of cardinality $n$ and let $L$ be a set of size $k$.
A family $F$ of hash functions from $U$ to $L$ is called \emph{perfect} if for
any subset $W \subseteq U$, such that $|W|=k$, there exists a function
$f \in F$ such that for each $x,y \in W$, $f(x)\neq f(y)$.
\end{Def}

Moreover, a perfect family $F$ of hash functions from $U$ to $\{l_1,
\dots, l_k \}$, having size $O(\log |U| 2^{O(k)})$, can be constructed
in time $O(2^{O(k)} |U| \log |U|)$ (see~\cite{ColorCoding}).

\section{An FPT Algorithm}
\label{sec:ftp1}
In this section, we present an FPT algorithm for Max-Duo PSM
parameterized by the number of preserved duos $k$ between the input
related strings $A$ and $B$.
More precisely, the FPT algorithm we present is parameterized by the
number of positions of $B$ that induce $k$ preserved duos.

Next, we prove the relation between the number of positions of $B$ inducing
the preserved duos and the number $k$ of preserved duos of a solution of Max-Duo
PSM on instance the $(A,B,k)$.

\begin{lemma}
\label{lem:pk}
Consider two related strings $A$ and $B$, a solution of Max-Duo PSM on the
instance $(A,B,k)$, and consider the $p$ positions of $B$ that induce such
$k$ duos.
Then it holds $p=k$.
\end{lemma}
\begin{proof}
The lemma follows easily from the fact that each preserved duo \DUO{S}{i} of a
string $S$, with $S \in \{ A,B \}$, is induced by position $i$ of $S$.
%
%
%
\end{proof}

Given an integer $k$, let $C = \{c_1, \dots , c_k \}$ be a set of $k$
\emph{colors}.
Let $F$ be a family of perfect hash functions from the positions of $B$ to the
set $C$.
Informally, we assign $k$ distinct colors to the positions of $B$ that
may induce preserved duos, and by dynamic programming we compute if
there exist $k$ distinct positions in $A$ that are mapped to
these candidate duos of $B$.
By Def.~\ref{def:perfect-hash}, we consider a function $f \in F$ that
associates a distinct color to each of the $k$ positions of $B$ that induces
a preserved duo.

Define $D[i,C']$, for $1 \leq i \leq n$, and $C' \subseteq C$, as a function
equal to $1$ if there   exist a set $S_B$ of $|C'|$ positions of $B$, each one
associated with a distinct color in $C'$, and a set $S_A$ of $|C'|$ positions
of $A[1,i]$, such that there exists a mapping from $S_A$ to $S_B$ that
preserves $|C'|$ duos; otherwise, the function is equal to $0$.

Define $P[h,i,C']$ as a function equal to $1$ when there exist positions $q$
and $r$ in $B$, with $1 \leq q < r \leq |B|$, such that each color in $C'$ is
associated with a position between $q$ and $r-1$, and substring $B[q,r]$ is
identical to $A[h,i]$; otherwise the function is equal to $0$.

We can compute $D[i,C']$ as follows:

\begin{equation*}
D[i,C'] = \max_{C''\subseteq C'}
\begin{cases}
D[h,C''] \times P[h+1,i,C'\setminus C'']\\
\qquad\qquad\text{where $h<i-1$, $i-h-1=|C'\setminus C''|$}\\
D[i-1,C']\\
\end{cases}
\end{equation*}

In the basic case it holds $D[1,C']=1$ if $|C'| = 0$, else $D[1,C']=0$.
It is easy to see that there exists a solution of Max-Duo PSM with $k$
preserved duos, if and only if $D[n,C]=1$.

Next we prove the correctness of the recurrence.

\begin{lemma}
\label{lem:corrColorCod}
Given two related strings $A$ and $B$, there exists a partial mapping of
$A[1,i]$ into $B$ that preserves $|C'|$ duos induced by positions of $B$
colored by $C'$ if and only if $D[i,C']=1$.
\end{lemma}
\begin{proof}
We prove the lemma by induction on $i$. First consider the basic case,
that is $i=1$, then $D[1,C']=1$ if and only if $|C'|=0$, since $A[1]$ contains
no duo.

Assume that the lemma holds for $j<i$, we show that it holds for $j=i$.

($\Rightarrow$) First, assume that $D[i,C']=1$, then we show that there exists
a partial mapping of $A[1,i]$ into $B$ that preserves $|C'|$ duos induced by
positions of $B$ colored by $C'$.

By assuming $D[i,C']=1$, then, if $D[i-1,C']=1$, by induction hypothesis there
exists a partial mapping of $A[1,i-1]$ into $B$ that preserves $|C'|$ duos
induced by positions  of $B$ colored by $C'$.
On the other hand, we have $D[i,C'] = D[h,C''] \times P[h+1,i,C'\setminus
C'']=1$, for some $C''\subseteq C'$, with $|C' \setminus C''|=i-h-1$.
Then, since $D[h,C'']=1$, by induction hypothesis there exists a partial
mapping of $A[1,h]$ into $B$ that preserves $|C''|$ duos induced by positions
of $B$ colored by $C''$.
Moreover, $P[h+1,i,C'\setminus C'']=1$, and it follows that there
exist positions $q$ and $r$ of $B$ such that each color in
$C'\setminus C''$ is associated with a distinct position $t$ of $B$,
with $q \leq t \leq r-1$, and $B[q,r]$ is identical to $A[h+1,i]$.
Hence, it follows that $i-h-1$ duos are preserved by mapping $A[h+1,i]$ into
$B[q,r]$, and are induced by positions of $B$ colored by $C' \setminus C''$.
As a consequence there exists a partial mapping of $A[1,i]$ into $B$
that preserves $|C'|$ duos induced by positions  of $B$ colored by $C'$.

($\Leftarrow$) Now, assume that there exists a partial mapping of $A[1,i]$ into
$B$ that preserves $|C'|$ duos induced by positions  of $B$ colored by $C'$.
We show that $D[i,C']=1$.

We can consider the following cases: there exists a sequence of preserved
consecutive duos \SQD{A}{h+1,i} mapped into a sequence of preserved
consecutive duos \SQD{B}{z+1,j}, with $j-z = i-h$  or no preserved duo is
induced by position $i-1$ of $A$.
In the latter case, by induction hypothesis, $D[i-1,C']=1$ and hence
$D[i,C']=1$.
In the former case, since function $f$ assigns a distinct color to each
position of $B$ that induces a preserved duo, there exists a set $C''$ such
that each position of \SQD{B}{z+1,j} inducing a preserved duo with
\SQD{A}{h+1,i} is associated with a distinct color in $C'' \subseteq C'$,
and each position of $B$ that induces a preserved duo with a position of
$A[1,h]$ is associated with a distinct color of  $C'\setminus C''$.
Hence, $P[h+1,i, C'']=1$, for some set $C'' \subseteq C'$.
Moreover, by induction hypothesis $D[h,C'\setminus C'']=1$, with
$i-h-1=|C' \setminus C''|$, and by the first case of the recurrence
$D[i,C']=1$.
\end{proof}

From the previous lemma, we can conclude the correctness of the algorithm.

\begin{teo}
\label{teo:ColorCoding}
Let $A$ and $B$ be two related strings on an alphabet $\Sigma$.
Then, it is possible to compute if there exists a solution of Max-Duo
PSM on instance $(A,B,k) $ in time $2^{O(k)}O(poly(n)\log(n))$.
\end{teo}
\begin{proof}
The correctness of the algorithm follows from the correctness of the
dynamic programming recurrence (see Lemma~\ref{lem:corrColorCod}).
Now, we consider the time complexity of the algorithm.
We recall that $n=|A|=|B|$.
First, assume that there exists a function $f$ 
in a perfect family $F$ of hash functions, such that $f$ color-codes the
positions of $B$.
Table $D[i,C']$ contains $O(2^k n)$ entries. Each entry is computed by
the recurrence which looks at most $2^k n$ possible entries.
Indeed, in the first case, the recurrence must check the entries $D[h,C'']$,
where $h \leq n$ (hence, there are at most $n$ of such values),  
and $C'' \subseteq C'$ (hence,
there are at most $2^k$ of such subsets). Moreover, the value of
$P[h,i,C' \setminus C'']$ can be checked in constant time.
Notice that $P[h,i,C' \setminus C'']$ can be precomputed in time
$O(2^k k^3 n^2)$.
Indeed  $P[h,i,C' \setminus C'']$ contains $O(2^k k n)$ positions, since for
each position $i$ in $A$, there exist at most $k$ positions $h+1$ and at most
$2^k$ subsets $C' \setminus C''$ to be checked.
Given positions $h$ and $i$, and subset $C' \setminus C''$, we must check
that each color in $C' \setminus C''$ is associated with a position between $q$
and $r-1$ of $B$, and that substring $B[q, r]$ is identical to $A[h , i]$.
This can be done in time $O(k^2 n)$ checking whether each substring $B[q, r]$
of length bounded by $k$ (there are at most $O(kn)$ of such strings) is identical to
$A[h, i]$ and each color in $C' \setminus C''$ is associated with a position
$j$ of $B[q, r-1]$.

It follows that table $D[i,C']$ can be computed in time $O(2^{2k} k^3 n^2)$
(considering the cost to precompute $P[h+1,i,C' \setminus C'']$).

In order to find an injective function $f$ in a perfect family $F$, we must iterate 
through the $2^{O(k)}O(\log(n))$ functions of $F$. Since the family $F$ can be computed in time
$2^{O(k)}poly(n)$ and $k \leq n$, it follows that the overall complexity is indeed $2^{O(k)}O(poly(n)\log(n))$.

%
\end{proof}


\section{A Reduction to a Polynomial Kernel}
\label{sec:kernel}

In this section, we prove that the Max-Duo PSM problem admits a
polynomial size kernel, by presenting a polynomial-time algorithm that,
starting from an instance $(A,B,k)$ of Max-Duo PSM, computes an instance
$(A',B',k)$, such that the length  of $A'$ and $B'$ is bounded by $O(k^6)$.

The general idea of the reduction is that in \emph{Phase 1}, starting from the related strings $A$ and $B$, we
compute two subsets of duos of $A$ and $B$, denoted by
$C_A$ and $C_B$ respectively, that may eventually be preserved, while any other duo not in these sets will not be preserved.
Then, in \emph{Phase 2}, starting from sets $C_A$ and $C_B$, we compute two related strings $A'$ and $B'$ respectively,
so that $(A',B',k) $ is an instance of Max-Duo PSM.

\subsection{Phase 1: Constructing Small Sets of Relevant Duos}

Here, we present the algorithm that in polynomial-time, starting from the related strings $A$ and $B$, 
computes two subsets $C_A$ and $C_B$, of duos of $A$ and $B$, respectively, called \emph{candidate sets}, having the
following properties:
\begin{enumerate}

\item there exists a solution of Max-Duo PSM on instance $(A,B,k)$ if and only if there exist  
$C'_A \subseteq C_A$, $C'_B \subseteq C_B$, with $|C'_A|=|C'_B|=k$ such that there is a mapping
of $C'_A$ into $C'_B$;

\item $C_A$ and $C_B$ contains $O(k^6)$ duos.
\end{enumerate}

In order to compute $C_A$ and $C_B$, the algorithm iteratively
adds (bounding its size) a set of duos
of $A$ (of $B$, respectively) to $C_A$ ($C_B$, respectively).
Recall that $k$ denotes the number of duos preserved by a solution
of Max-Duo PSM.


Before giving the details, we describe informally the three rules on which our
Phase 1 of the kernelization is based.
Rule 1 computes a maximum matching $M$ of a graph that represents the duos of $A$ and
$B$. Since $M$ is a maximum matching, we are sure that if there exists a preserved
duo whose corresponding string is $'ab'$ in a solution, then 
all the duos whose corresponding string is $'ab'$ will be added
to set $C_A$ or set $C_B$ (see Lemma~\ref{lem:bound:match2}).

Consider one of such duos \DUO{A}{i}; in Rule 2, we add to $C_A$ all the duos (not already included in $C_A$) that
belong to the sequences of consecutive duos induced by the substrings of size
$k+1$ that include \DUO{A}{i}.
This ensures that if a sequence \SQD{A}{h,j}, that consists of $k$ duos and
that includes \DUO{A}{i}, is preserved by a solution of Max-Duo PSM, then the
duos of \SQD{A}{h,j} are included in $C_A$.
Finally, notice that \DUO{A}{i} can be mapped to a set of duos of $B$ 
whose size is not bounded by a function of $k$.
Rule 3 adds either all the sequences of duos of $B$ where a sequence of at most
$k$ duos of $C_A$ can be mapped (if such sequences are at most $k^2$) or it
adds at most $k^2+1$ of such sequences (without loss of generality we add the
$k^2+1$ leftmost of such sequences).
It is sufficient to add $k^2+1$ sequences of duos, since we can ensure (see
Lemma~\ref{lem:bound:extStr2} and Lemma~\ref{lem:bound:extStr3}) that if there
exists a solution where \DUO{A}{i} is preserved, then the same property holds
for a partial mapping of the duos of $C_A$ into the duos of $C_B$.
Indeed, since there exist at most $k$ sequences of preserved duos of $A$,
each one having length at most $k$, it follows that they can overlap at most
$k^2$ sequences of duos where \DUO{A}{i} can be mapped.
For this reason, by adding $k^2+1$ disjoint sequences of duos, we can guarantee
that there exists a sequence of duos of $B$ where \DUO{A}{i} can be mapped.

%
We start by giving the details of our algorithm.
First, we consider an easy bound on the length of each sequence of
consecutive duos of $A$ and $B$ that can be preserved. Notice indeed
that if there exists a sequence of consecutive duos of $A$ having length
at least $k$ that can be mapped into a sequence of consecutive duos of $B$ having
length $k$ (which can be computed in polynomial time), then obviously there exists
a solution that preserves at least $k$ duos. Hence, we assume that the following
claim holds.

\begin{Cl}
\label{claim:boundlength}
There is no sequence of consecutive duos of $A$ having length at least $k$ that can be mapped into 
a sequence of consecutive duos of $B$ having length at least $k$.
\end{Cl}

Now, we are able to define the rules for the Phase 1 of the kernelization.
The first rule is based on the approach of~\cite{DBLP:conf/wabi/BoriaKLM14} 
that leads to a $\frac{1}{4}$-approximation algorithm.
The approximation algorithm given in~\cite{DBLP:conf/wabi/BoriaKLM14}  is
based on a graph representation of the duos of the given input strings. 
A maximum matching of this graph is then computed and it is decomposed into
four submatchings; the maximum of such submatchings is then returned as the approximated solution
of factor $\frac{1}{4}$.

As in~\cite{DBLP:conf/wabi/BoriaKLM14}, we first consider a bipartite
graph $G=(V_A \uplus V_B, E)$ associated with the related strings $A$
and $B$, input of  Max-Duo PSM, and defined as follows:
\begin{itemize}

\item for each duo in $A$, there exists a vertex in $V_A$;

\item for each duo in $B$, there exists a vertex in $V_B$;

\item there exists an edge $\{v_a, v_b\}\in E$ connecting a vertex $v_a \in
V_A$ to a vertex $v_b \in V_B$ if and only if they represent a
preservable duo.

\end{itemize}

Now, we are ready to present the first rule of the kernelization algorithm.

\paragraph*{\textbf{Rule 1}}
Compute (in polynomial time) a maximum matching $M\subseteq E$ of $G$
and define $C_A$ and $C_B$ as the sets of duos corresponding to
the endpoints of each edge of $M$. More precisely
\[
C_A = \{ v_a \in V_A | \{v_a, x\} \in M \}
\]
and
\[
C_B = \{ v_b \in V_B | \{x, v_b\} \in M \}.
\]

It can be shown that $|C_A|, |C_B| \leq 4k$, since otherwise 
we can compute a solution of Max-Duo PSM on instance $(A,B,k)$.

\begin{lemma}
\label{lem:bound:match1}
Given two related strings $A$ and $B$, let $G$ be the corresponding graph. 
Let $M$ be a maximum matching of $G$ and let $C_A$ and $C_B$ be the two sets
of duos built by Rule 1. Then, if $|M| \geq 4k$ and $|C_A|, |C_B| \geq 4k$,
Max-Duo PSM on  instance $(A,B,k)$ admits a feasible solution.
\end{lemma}
\begin{proof}
It is shown in~\cite{DBLP:conf/wabi/BoriaKLM14} that the value of a maximum matching $M$
is an upper bound on the number of preserved duos of the related strings $A$ and $B$.
Moreover, in~\cite{DBLP:conf/wabi/BoriaKLM14} it is shown that $M$
can be partitioned in polynomial-time into four submatchings, such that each of them induces a partial mapping
of $A$ into $B$; if more than $k$ duos are preserved
by one of the partial mappings, then it is a solution of Max-Duo PSM on
instance $(A,B,k)$.
Then, if $|M| \geq 4k$ and, by construction of $C_A$ and $C_B$, $|C_A|, |C_B| \geq 4k$, 
one of the submatchings induces a partial mapping
of $A$ into $B$ that preserves at least $k$ duos, hence 
Max-Duo PSM on  instance $(A,B,k)$ admits a feasible solution.
\end{proof}

In the following, we assume that $|M| < 4k$ and $|C_A|, |C_B| < 4k$.
Next, we prove another useful property of the computed maximum matching $M$ of $G$.
We denote by $M_A$ ($M_B$, respectively) the set of vertices of $V_A$ ($V_B$,
respectively) that are endpoints of an edge belonging to $M$.

\begin{lemma}
\label{lem:bound:match2}
Consider the symbols $a,b \in \Sigma$ and assume that there exist preservable
duos of the related strings $A$ and $B$, whose corresponding string is $'ab'$.
Then, at most one of the sets $V_A \setminus M_A$ and $V_B \setminus M_B$
contains a vertex associated with a duo whose corresponding string is $'ab'$.
\end{lemma}
\begin{proof}
Assume by contradiction that the lemma does not hold.
It follows that there exists one vertex of $v_a \in V_A$ associated with a duo of $A$ whose corresponding string is $'ab'$
and there exists one vertex of $v_b \in V_B$ associated with a duo of $B$ whose corresponding string is $'ab'$, 
such that $v_a$, $v_b$ are not endpoints of an edge of $M$.
Hence by adding such an edge to $M$ (which exists by construction of
$G$), it is possible to obtain a matching larger than $M$, which
contradicts the fact that $M$ is maximum.
\end{proof}

%

We recall that, given two positions $1 \leq i < j \leq n$, \SQD{S}{i,j} denotes 
the sequence of \emph{consecutive  duos} \DUO{S}{i}, \dots , $(S[j-1],S[j])$. 
By a slight abuse of notation, we denote by \SQD{S}{i-k,i+k} the sequence
of duos between position $S[l]$, where $l=\max\{1,i-k\}$, and position 
$S[r]$, where $r=\min\{n,i+k\}$.


\paragraph*{\textbf{Rule 2}} 
For each duo \DUO{S}{i} of $C_{S}$, with $S \in \{ A,B \}$,
%
%
add to $C_S$ all the duos of \SQD{S}{i-k,i+k}.

\vspace*{.3cm}
Given $S\in \{A, B\}$, we recall that $\bar{S}$ is the string in $\{A, B\}
\setminus \{S\}$. It can be shown that the following properties hold.
\begin{lemma}
\label{lem:bound:extStr}
Given a string $S$, with $S \in \{ A,B \}$, consider a duo \DUO{S}{i} added by Rule 1 to $C_{S}$.
If there exists a solution of Max-Duo PSM on instance $(A,B,k)$ that maps a sequence \SQD{S}{i-t_1,i+t_2} 
with $0 \leq t_1 \leq k$ and $1 \leq t_2 \leq k$, of consecutive duos of $S$ that includes \DUO{S}{i} into a 
sequence \SQD{\bar{S}}{i-u_1,i+u_2} of consecutive duos of $\bar{S}$, then Rule 2 adds all the
duos of \SQD{S}{i-t_1,i+t_2}
to $C_S$.
\end{lemma}
\begin{proof}
From Claim~\ref{claim:boundlength}, it follows that 
we assume that if a solution of Max-Duo PSM on instance $(A,B,k)$ defines a mapping of a sequence 
\SQD{S}{i-t_1,i+t_2} of $S$ into a sequence \SQD{\bar{S}}{i-u_1,i+u_2} of $\bar{S}$, then $t_2+t_1 +1  = u_2+u_1 +1\leq k$,
hence $t_1,t_2 \leq k$ and $u_1, u_2\leq k$. Since Rule 2 adds to $C_S$ the sequence \SQD{S}{i-k,i+k} of duos,
the lemma holds.
%
\end{proof}

Moreover, we can bound the number of duos added by Rule 2 as follows.

\begin{lemma}
\label{lem:bound:number:extrStr}
Rule 2 adds at most $8k^2$ duos to each set $C_S$, with $S \in \{ A,B
\}$.
\end{lemma}
\begin{proof}
Since $|M|< 4k$, it follows that there exist less than $4k$ positions $i$, with $1 \leq i \leq n$, such that 
the sequence \SQD{S}{i-k,i+k} of duos are added to $C_S$.
Since there exist $2k$ consecutive duos in \SQD{S}{i-k,i+k},
at most $8k^2$ duos are added to $C_S$.
\end{proof}

%

We are now able to define Rule 3.

\paragraph*{\textbf{Rule 3}}
Consider a sequence \SQD{S}{p,q} of consecutive duos that has length at most $k$, 
such that each duo \DUO{S}{i}, with $p \leq i \leq q-1$, is added by Rule 1 and Rule 2 to set $C_S$;
add a set of
candidate duos to $C_{\bar{S}}$ as follows:
\begin{itemize}
\item if there exist at least $k^2+1$ non-overlapping sequences in
$\bar{S}$ where  \SQD{S}{p,q} can be mapped:
add to
$C_{\bar{S}}$ all the duos belonging to the leftmost non-overlapping
$k^2+1$ sequences in $\bar{S}$ where
\SQD{S}{p,q} can be mapped;

\item else, add to $C_{\bar{S}}$ all the duos belonging to the
sequences of consecutive duos in $\bar{S}$ where
\SQD{S}{p,q} can be mapped.
\end{itemize}

It can be shown that the following property holds.
\begin{lemma}
\label{lem:bound:extStr2}
Consider a solution $X$ that preserves a sequence \SQD{S}{p,q}
of consecutive duos such that each duo \DUO{S}{i}, with $p \leq i \leq q-1$, is added by Rule 1 and Rule 2 to sets $C_S$.
%
Then Rule 3 either adds to $C_{\bar{S}}$ the duos of $\bar{S}$ where
\SQD{S}{p,q}
is mapped by $X$, or there are $k^2+1$ non-overlapping
sequences of consecutive duos of $\bar{S}$ in $C_{\bar{S}}$ where 
\SQD{S}{p,q} can be mapped.
\end{lemma}
\begin{proof}
Consider a sequence \SQD{S}{p,q}
of consecutive duos preserved by a
solution $X$ of Max-Duo PSM on instance $(A, B,k)$.
Then, consider the two cases of Rule 3. In the second case, since
$C_{\bar{S}}$ contains all the consecutive duos in $\bar{S}$ where
\SQD{S}{p,q} can be mapped,
then the lemma holds.
Now, consider the first case of Rule 3. Then, by construction in $C_{\bar{S}}$ there
are $k^2+1$ non-overlapping sequences of consecutive duos in $\bar{S}$
where \SQD{S}{p,q} can be mapped.

\end{proof}

\begin{lemma}
\label{lem:bound:extStr3}
Consider a solution $X$ of Max-Duo PSM on instance $(A, B,k)$. 
Then, there exist subsets of duos $C'_A \subseteq C_A$ and $C'_B \subseteq C_B$,
such that there is a mapping of $C'_A$ into $C'_B$ and $|C'_A|=|C'_B|=k$.
\end{lemma}
\begin{proof}

Consider a duo \DUO{S}{i}, where $S\in\{A, B\}$, $S[i]=a$
and $S[i+1]=b$, such that a solution $X$ of Max-Duo PSM on instance $(A, B,k)$
maps a sequence \SQD{S}{p,q} of consecutive duos of $S$,
with $i-k \leq p < q \leq i+k$, into a sequence of consecutive duos 
\SQD{\bar{S}}{t,u}. 
Notice that all the duos of \SQD{S}{p,q} belong to $C_S$, or all the duos of \SQD{\bar{S}}{t,u}
belong to $\bar{S}$. Indeed, assume that the duo \DUO{S}{i} is mapped into the duo 
\DUO{\bar{S}}{j} by $X$. By Lemma~\ref{lem:bound:match2}, it follows that Rule 1 adds 
duo \DUO{S}{i} to $C_S$ or \DUO{\bar{S}}{j+1} to $C_{\bar{S}}$. Moreover,
by Lemma~\ref{lem:bound:extStr}, Rule 2 adds all the duos of \SQD{S}{p,q} to $C_S$ or all the duos of
\SQD{\bar{S}}{t,u} to $\bar{S}$.
In what follows, we assume w.l.o.g. that all the duos \SQD{S}{p,q} belong to $C_S$.

Let $S$ be the string having the maximum number of occurrences of
duos whose corresponding string is $'ab'$.
Then, by Lemma~\ref{lem:bound:match2} and by Rule 1, each duo of
$\bar{S}$, whose corresponding string is $'ab'$, is added to $C_{\bar{S}}$.
Moreover, consider a sequence \SQD{S}{p,q}
of consecutive duos of $S$,
with $i-k \leq p < q \leq i+k$, mapped into a sequence
\SQD{\bar{S}}{t,u}
of consecutive duos of $\bar{S}$. By
Lemma~\ref{lem:bound:extStr}, all the duos of \SQD{\bar{S}}{t,u}
are added to $C_{\bar{S}}$ by Rule 2.
Hence, we can assume that the duos of \SQD{S}{p,q} and the duos of \SQD{\bar{S}}{t,u}
belong to $C'_A$ and $C'_B$ respectively, and there is a mapping between such duos.

Assume that is $S$ the string having the minimum number of
occurrences of $'ab'$ 
%
and consider the sequence \SQD{S}{p,q}
of consecutive duos of $X$ mapped into the sequence \SQD{\bar{S}}{t,u} of
consecutive duos. 
By Lemma~\ref{lem:bound:extStr2}, either all the consecutive duos of
$\bar{S}$ where \SQD{S}{p,q}
can be mapped are included in $C_{\bar{S}}$,
or the duos of $k^2+1$ non-overlapping sequences of consecutive
duos of $\bar{S}$, where \SQD{S}{p,q}
can be mapped, are included in $C_{\bar{S}}$.

In the former case, all the duos of \SQD{\bar{S}}{t,u} are added to $C_{\bar{S}}$ by Rule 3,
hence we can assume that the duos of \SQD{S}{p,q} and the duos of \SQD{\bar{S}}{t,u}
belong to $C'_S$ and $C'_{\bar{S}}$ respectively, and there is a mapping between such duos. 

%
%

In the latter case, 
%
we show that, even if $C_{\bar{S}}$ does not contain
all the duos of \SQD{\bar{S}}{t,u},
it contains a sequence of
consecutive duos where \SQD{S}{p,q}
can be mapped.
Notice that each sequence of consecutive duos of $\bar{S}$ mapped to 
sequence of consecutive duos of $S$  has length bounded by
$k$ (see Claim~\ref{claim:boundlength}) and, since $X$ preserves $k$ duos, the set $C'_{\bar{S}}$
contains at most $k$ of such sequences of consecutive duos.
It follows that each sequence of consecutive duos added to $C'_{\bar{S}}$ can overlap
at most $k$ non-overlapping occurrences of \SQD{S}{p,q} in $\bar{S}$. Hence the whole set
of sequences of consecutive duos preserved of $C'_{\bar{S}}$ can overlap at most $k^2$
non-overlapping sequence of consecutive duos in $\bar{S}$ where
\SQD{S}{p,q} can be mapped.
%
Since we have added to $C_{\bar{S}}$ the duos belonging to $k^2+1$
non-overlapping sequences of consecutive duos of $\bar{S}$ 
where \SQD{S}{p,q} can be mapped,
then there exists at least a sequence \SQD{\bar{S}}{w,y}
of consecutive duos in $C_{\bar{S}}$ where \SQD{S}{p,q}
can be mapped.
Hence we can assume that the duos of \SQD{S}{p,q} and the duos of \SQD{\bar{S}}{w,y}
belong to $C'_S$ and $C'_{\bar{S}}$ respectively, and there is mapping between such duos. 
\end{proof}

Now, we are able to bound the size of the sets $C_A$ and $C_B$.

\begin{lemma}
\label{lem:bound:number:extrStr2}
Given an input string $S \in \{A,B\}$, Rules 1-3 add at most $O(k^6)$ duos to each set $C_S$.
\end{lemma}
\begin{proof}
By Lemma~\ref{lem:bound:match2}, Rule 1 adds at most $4k$ duos to each set $C_S$, with $S \in \{A,B\}$.
By Lemma~\ref{lem:bound:number:extrStr}, Rule 2 adds at most $O(k^2)$
duos to each set $C_{S}$, with $S \in \{ A,B \}$. 

Rule 3 considers
$O(k^2)$ sequences \SQD{S}{p,q} of length bounded by $k$.
%
For each such sequences of consecutive duos,
Rule 3 adds duos to $C_{\bar{S}}$ in two possible ways.
In the first case $k^2+1$ sequences of consecutive duos are selected,
each one having size at most $k$, thus $O(k^3)$ duos are added to $C_{\bar{S}}$; thus
$O(k^5)$ are added to $C_{\bar{S}}$.
In the second case there exist at most $k^2$ non-overlapping sequences
of consecutive duos (each one having length at most $k$) that can preserve \SQD{S}{p,q}.
Moreover, each of the $O(k^2)$ non-overlapping sequences of consecutive duos
where \SQD{S}{p,q} can be mapped
can overlap at most $k$ sequences of consecutive duos that can preserve
\SQD{S}{p,q}. Thus, for each of such $O(k^2)$ non-overlapping sequences of consecutive duos,
$O(k^4)$ duos are added to $C_{\bar{S}}$.
%
Hence, the overall number of consecutive duos added to
$C_{\bar{S}}$ is $O(k^6)$.
\end{proof}

\subsection{Phase 2: Completing the construction}
From the sets $C_A$ and $C_B$ previously computed, we construct an instance
of Max-Duo PSM, that is two related strings $A'$ and $B'$.
Furthermore, we will show that the length of $A'$ and $B'$ is bounded by $O(k^6)$ and that the preservable
duos of $A'$ and $B'$ are those of $C_A$ and $C_B$, respectively.

Recall that $A$ and $B$ are two related strings over alphabet $\Sigma$. 
Consider the set $T_A$ of substrings of $A$ (the set $T_B$ of substrings of $B$, respectively) that induces
the duos in $C_A$ (in $C_B$, respectively) and assume that $T_S$, with $S \in \{ A,B\}$,
contains $q_S$ strings, namely $t_{1,S}, \dots t_{q_S,S}$.

Before describing the two strings
$A'$ and $B'$, we construct the alphabet $\Sigma'$ on which they are based, where:
\begin{align*}
\Sigma'=&\enskip\Sigma\enskip \cup\\
	&\enskip\{ e_{A,i}: 1 \leq i \leq q_A  \}\enskip \cup\\
	&\enskip\{ e_{B,i}: 1 \leq i \leq q_B  \}\enskip \cup\\
	&\enskip\{g_a:\text{$a \in \Sigma$ has a different number of occurrences in $T_A$ and $T_B$} \} 
\end{align*}

We concatenate the substrings of $T_A$ (of $T_B$, respectively) with symbols of
$\Sigma' \setminus \Sigma$.
%
We compute an intermediate string $P_S$, with $S \in \{ A,B\}$, as follows.
First, set $P_S=t_{1,S}\cdot e_{S,1}$. Then, for each $i$ with
$2 \leq i \leq q_S$, concatenate the string $P_S$ with string $t_{i+1,S}\cdot e_{S,i}$.
Finally, append string $e_{\bar{S},1} \dots e_{\bar{S},q_S}$ at the right
end of $P_S$.


Now, we append some strings to the right end of $P_A$ and $P_B$ in order to
compute $A'$ and $B'$ over alphabet $\Sigma'$.
More precisely, for each symbol $a \in \Sigma$ such that the number of occurrences
of $a$ in $P_A$  and in $P_B$ is different, we apply the following procedure. 
Assume w.l.o.g. that $P_A$ contains $h_{A,a}$ occurrences of symbol $a$ and 
$P_B$ contains $h_{B,a}$ occurrences of symbol $a$, with $h_{A,a} >  h_{B,a}$.
Then, we append a string $(g_a a) ^{h_{A,a} -  h_{B,a}}$ (that is the string consisting
of  the concatenation of $h_{A,a} -  h_{B,a}$ occurrences of $g_a a$) to the right end
of $P_B$, while we append the string  $(g_a) ^{h_{A,a} -  h_{B,a}}$ (that is the string
consisting of  the concatenation of $h_{A,a} -  h_{B,a}$ occurrences of $g_a$) to the
right end of $P_A$. 
Similarly, if $h_{A,a} <  h_{B,a}$, we append a string $( g_a a) ^{h_{B,a} -  h_{A,a}}$
to the right end of $P_A$, while we append the string  $(g_a) ^{h_{B,a} -  h_{A,a}}$
to the right end of $P_B$ (see Figure~\ref{fig:phase2}).

\begin{figure}[ht]
\centering
\begin{align*}
A'=&
\overbrace{t_{1,A}\cdot e_{A,1} \dots t_{q_{A},A}\cdot e_{A,q_{A}}}^{P_{A}}
e_{B,1} \dots e_{B,q_{B}}
\overbrace{g_{a} g_{a} \dots g_{a}}^{(g_{a})^{h_{A,a}-h_{B,a}}}
\overbrace{g_{b} g_{b} \dots g_{b}}^{(g_{b})^{h_{A,b}-h_{B,b}}}
\overbrace{g_{c} c g_{c} c \dots g_{c} c}^{(g_{c}c)^{h_{B,c}-h_{A,c}}}
\\
B'=&
\underbrace{t_{1,B}\cdot e_{B,1} \dots t_{q_{B},B}\cdot e_{B,q_{B}}}_{P_{B}}
e_{A,1} \dots e_{A,q_{A}}
\underbrace{g_{a} a g_{a} a \dots g_{a} a}_{(g_{a}a)^{h_{A,a}-h_{B,a}}}
\underbrace{g_{b} b g_{b} b \dots g_{b} b}_{(g_{b}b)^{h_{A,b}-h_{B,b}}}
\underbrace{g_{c} g_{c} \dots g_{c}}_{(g_{c})^{h_{B,c}-h_{A,c}}}
\end{align*}
\caption{An example of two related strings $A'$ and $B'$ constructed during
Phase 2 from sets $C_{A}$ and $C_{B}$ (and sets $T_A$and $T_B$ of substrings),
respectively.
We assume that the string $T_A$ contains more occurrences of symbol $a$, $b$
(denoted by $h_{A,a}$ and $h_{A,b}$, respectively) than $T_B$ (denoted by
$h_{B,a}$ and $h_{B,b}$, respectively), while string $T_B$ contains more
occurrences of symbol $c$ (denoted by $h_{B,c}$) than $T_A$ (denoted by
$h_{A,c}$).
Hence, strings $g_{a}a \dots g_{a}a $ and $g_{b}b \dots g_{b}b$ are appended
to construct $A'$, while strings $g_{a} \dots g_{a}$ and $g_{b} \dots g_{b}$
are appended to construct $B'$; string $g_{c}c \dots g_{c}c$ is appended to
construct $B'$, while string  $g_{c} \dots g_{c}c $ is appended to construct
$A'$.
}
\label{fig:phase2}
\end{figure}

The following lemma guarantees that the instance built on $A'$ and $B'$,
is an instance of Max-Duo PSM.

\begin{lemma}
\label{lem:kernelPadding1}
Let $A'$ and $B'$ be the two strings computed starting from $T_A$ and $T_B$,
respectively.
Then, $A'$ and $B'$ are related.
Moreover, $|A'|$ and $|B'|$ are bounded by $O(k^6)$.
\end{lemma}
\begin{proof}
Notice that each symbol in $\Sigma' \setminus \Sigma$ by construction has the
same number of occurrences in $A'$ and $B'$.
Moreover, for each symbol in $a \in \Sigma$, both $A'$ and $B'$ contain by construction
$\max(h_{A,a},h_{B,a})$ occurrences of $a$.

From Lemma~\ref{lem:bound:number:extrStr2} it follows that $|C_A|$ and $|C_B|$ are bounded by $O(k^6)$,
hence the symbols $e_{S,i}$, with $1\leq i\leq q_S$ (where $S\in \{A,B\}$),
inserted in $A'$ and $B'$ are at most $O(k^6)$.
Similarly, the number of symbols $g_a \notin \Sigma$ and $a \in \Sigma$ inserted are at most
$O(k^6)$, since in the worst case at most $|C_A|$ ($|C_B|$, respectively) symbols of $\Sigma$
are inserted in $A'$ (in $B'$, respectively), and for each of such symbols, exactly one occurrence
of some symbol $g_a$ is inserted in $A'$ and $B'$. 
\end{proof}

Finally, the following property holds.

\begin{lemma}
\label{lem:kernelPadding2}
Let $X$ be a solution of Max-Duo PSM on instance $(A',B',k)$, 
such that 
$X$ maps \DUO{A'}{i} into \DUO{B'}{j}.
Then 
\DUO{A'}{i} and \DUO{B'}{j} are duos of $C_A$ and $C_B$, respectively.
\end{lemma}
\begin{proof}
The lemma follows from the fact that, by construction, each preservable duo of
$A'$ and $B'$ belongs to strings in $T_A$ and $T_B$, respectively. 
Indeed, consider w.l.o.g.~a duo \DUO{S'}{i}  of $S'$, with $S' \in \{ A', B' \}$,
not in $C_S$. 
Then, by construction of $S'$, \DUO{S'}{i}  must include at least a symbol $x$
in $\Sigma' \setminus \Sigma$, as each new occurrence of a symbol in $\Sigma$ 
is adjacent to a symbol in $\Sigma' \setminus \Sigma$. 
Now, notice that if $x$ is equal to $e_{S,j_S}$, with $1 \leq j \leq q_S$, then by construction it belongs
to duos with different symbols in $S'$ and $\bar{S}'$, as $x$ is adjacent in
$\bar{S}'$ only to symbol in  $\Sigma' \setminus \Sigma$, while in $S$ is adjacent
only to symbols in $\Sigma$, with the exception of $e_{S,q_S}$, which again by
construction cannot belong to a preservable duo.

Now, consider a symbol $g_y \in \Sigma' \setminus \Sigma$.
Then exactly one of $A'$, $B'$ contains a duo whose corresponding string is $y g_y$ and exactly one of $A'$, $B'$
contains a duo whose corresponding string is $g_y g_y$.
Moreover, exactly one of $A'$, $B'$ contains a duo whose corresponding string is $z g_y$ and exactly one  
of $A'$, $B'$ contains a duo whose corresponding string is $g_z g_y$, with $z$ a symbol in $\Sigma$, with the exception
of the first symbol appended to $P_A$ or $P_B$, which again by construction
cannot belong to a preservable duo.
%
%
%
%
%
\end{proof}

We conclude the description of the kernelization algorithm with the following theorem.

\begin{teo}
\label{teo:Kern}
Given an instance $(A,B,k)$ of Max-Duo PSM, Phase 1 and Phase 2 compute in time $O(n^{\frac{5}{2}} k^6)$ 
an instance $(A',B',k)$, with $|A'|$ and $|B'|$ bounded by $O(k^6)$,
such that there exists a solution of Max-Duo PSM on instance $(A,B,k)$ if and only if there exists a solution
of Max-Duo PSM on instance $(A',B',k)$.
\end{teo}
\begin{proof}
Notice that by Lemma~\ref{lem:kernelPadding1}, $|A'|$ and $|B'|$ are bounded by $O(k^6)$.
First, we show that there exists a solution of Max-Duo PSM on instance $(A,B,k)$ if and only if there exists a solution
of Max-Duo PSM on instance $(A',B',k)$.

($\Rightarrow$) Consider a solution of Max-Duo PSM on instance $(A',B',k)$, then a solution Max-Duo PSM on instance $(A,B,k)$
can be computed by preserving those duos of $(A,B,k)$ corresponding to the duos preserved by the solution
of Max-Duo PSM on instance $(A',B',k)$.

($\Leftarrow$) By Lemma~\ref{lem:bound:extStr3} and by Lemma~\ref{lem:kernelPadding2}, 
if there exists a solution of of Max-Duo PSM on instance $(A,B,k)$, then there exists a solution
of of Max-Duo PSM on instance $(A',B',k)$.

Now, we show that Phase 1 and Phase 2 compute instance $(A',B',k)$ in time $O(n^{\frac{5}{2}} k^6)$.
Rule 1 requires time $O(n^2)$ to compute the graph $G$ and $O(n^{\frac{5}{2}})$ to compute 
a maximum bipartite matching of $G$~\cite{DBLP:journals/siamcomp/HopcroftK73}.
Rule 2 requires time $O(k^2)$ to add to $C_A$ and $C_B$, for each of the at most $O(k)$ duos of $M_A$ and $M_B$, the $O(k)$ duos
of \SQD{S}{i-k,i+k}.
Rule 3 requires time $O(k^6 n)$ to add to $C_A$ and $C_B$, for each duo added by Rule 1-2 (which are at most $O(k^2)$), 
the sequences \SQD{\bar{S}}{p,q} (which are at most $O(k^3)$).

Finally, Phase 2 computes $A'$ and $B'$ in time $O(k^6)$. Indeed, $P_S$ can be computed in time $O(k^6)$ by concatenating
a set of $O(k^6)$ strings. Moreover, we can count the number of occurrences of symbols of $\Sigma$ in a string
$P_S \in \{ A,B \}$ in time $O(k^6)$, and we append $O(k^6)$ strings $g_a a$, $g_a$ in time $O(k^6)$.
\end{proof}

\section{Conclusion}

In this paper, we have investigated the parameterized complexity of the
Max-Duo PSM problem, by first giving a parameterized algorithm based
on color-coding and then showing that it admits a kernel of size
$O(k^6)$.
From a paramterized complexity point of view, there are some interesting open problems for Max-Duo PSM. 
First, following the approach of \emph{parameterizing above a guaranteed value}, 
it would be interesting to investigate the parameterized complexity of the problem when the parameter is the number
of conserved duos minus the conserved duos induced by the submatching returned by the approximation algorithm in~\cite{DBLP:conf/wabi/BoriaKLM14}.
Furthermore, it would be interesting to improve upon the time (and space) complexity
of our color-coding based algorithm. In particular, notice that our color-coding based algorithm requires exponential space complexity, 
as it makes use of two tables of size $O(2^k k n)$.


%

\bibliographystyle{splncs03}
\bibliography{biblio}

\end{document}